\documentclass[a4paper,10pt]{article}

\usepackage{amsmath}
\usepackage{amssymb}
\usepackage{amsthm}
\usepackage{graphicx}
\usepackage{authblk}
\usepackage{subfigure}
\usepackage{caption}
\usepackage[left=3cm, right=3cm, top=3cm, bottom=3cm]{geometry}

\renewcommand{\leq}{\leqslant}
\renewcommand{\geq}{\geqslant}

\newcommand{\C}{\mathbb{C}}
\newcommand{\E}{\mathbb{E}}

\newcommand{\U}{\mathcal{U}}
\newcommand{\D}{\mathcal{D}}
\newcommand{\M}{\mathcal{M}}

\renewcommand{\S}{\mathcal{S}}
\newcommand{\V}{\mathcal{V}}
\newcommand{\ol}{\overline}

\newcommand{\floor}[1]{\lfloor #1 \rfloor}

\DeclareMathOperator{\trace}{Tr}

\DeclareMathOperator{\I}{I}
\DeclareMathOperator{\id}{id}
\DeclareMathOperator{\Wg}{Wg}
\DeclareMathOperator{\Mob}{Mob}

\DeclareMathOperator{\Rem}{Rem}
\newcommand{\F}{\mathcal{F}}

\renewcommand{\phi}{\varphi}
\newcommand{\iy}{\infty}

\newcommand{\gammat}{\tilde{\gamma}}
\newcommand{\fhat}{\hat{f}}

\newcommand{\ketbra}[2]{| #1 \rangle \langle #2 |}

\newtheorem{theorem}{Theorem}[section]
\newtheorem{definition}[theorem]{Definition}
\newtheorem{proposition}[theorem]{Proposition}
\newtheorem{remark}[theorem]{Remark}

\title{Eigenvalue and Entropy Statistics for Products of Conjugate Random Quantum Channels}
\author[,1,2]{Beno\^it Collins\footnote{bcollins@uottawa.ca}}
\author[,1]{Ion Nechita\footnote{inechita@uottawa.ca}}
\affil[1]{Dept. of Mathematics and Statistics, University of Ottawa, ON, Canada}
\affil[2]{CNRS, Institut Camille Jordan, Universit\'e  Lyon 1, France}
\date{}
\begin{document}
\maketitle
\abstract{Using the graphical calculus and integration techniques introduced by the authors,
we study the statistical properties 
of outputs of products of random quantum channels for entangled inputs. In particular, we revisit and generalize models of relevance 
for the recent counterexamples to the minimum output entropy additivity problems. 
Our main result is a classification of regimes for which the von Neumann entropy is lower on average than the elementary bounds that can 
be obtained with linear algebra techniques.
}

\section{Introduction}

As in classical computer science, randomized proofs and constructions are ubiquitous in quantum information. Since quantum mechanics is non commutative, the random objects of study are matrices. Therefore, quantum information theory provides a rich source of random matrix problems.

One of the most important classes of problems in the mathematical aspects of quantum information theory is the study of data transmission through noisy quantum channels. A famous conjecture reduced the calculation of the channel capacity for classical data to the question of the additivity of Minimum channel Output Entropy. The conjecture was stated in 1999 by King and Ruskai \cite{king-ruskai} and shown to be equivalent to the additivity of the Holevo capacity (and to other quantities of interest) by Shor \cite{shor}. For a long time, no counterexamples were available and additivity was proven to hold in many cases. A stronger, $L^p$ version of this question was also available and relevant to operator algebra and operator space theory. This version was disproved by Hayden and Winter in 2007, for all $p>1$ \cite{hayden-winter}.
The original conjecture, regarding von Neumann entropies, was disproved by Hastings in 2008 \cite{hastings}. His very innovative argument exploited the idea of tubular neighborhoods so as to considerably refine available estimates on random quantum channels. However, constructive, non-random counterexamples to any of these conjectures are still elusive. 

The random counterexamples to the various forms of the additivity conjecture rely on bounds on the Minimum Output Entropies (MOE) for single and product channels that follow mainly from two important ideas. Let $\Phi$ be a random quantum channel between matrix spaces such that the dimensions of the input and output spaces are large enough. The first key idea is that, with high probability, the Minimum Output Entropy of $\Phi$ is almost maximal: all output states are highly mixed. On the other hand, if one considers the product channel $\Phi\otimes \overline\Phi$ (where $\overline\Phi$ is obtained by replacing the Stinespring unitary $U$ defining the channel by its conjugate), then if one takes a maximally entangled (or Bell) state as an input, the output density matrix has \emph{always} a large eigenvalue. This second important fact was observed by Winter, and it implies that the output state in question has low entropy, allowing for a violation of additivity.

Our work addresses bounds for conjugate product channels, and improves them in several cases. We provide a complete spectral description of the output of product channels when the input is maximally entangled. In \cite{collins-nechita-1} and \cite{collins-nechita-3}, we have studied situations when the channels are conjugate ($\Phi \otimes \ol \Phi$) or independent ($\Phi \otimes \Psi$) in two different asymptotic regimes. 

In this work, after recalling the aforementioned results and reviewing the techniques used in deriving them, we consider more general models of random quantum channels, from two different perspectives. We first generalize the linear scaling asymptotic regime to include the situations where the dimension of the input space is different from the dimension of the output of a quantum channel; however, all three parameters (the respective dimensions of the input, output and ancilla spaces) scale linearly. Then we move beyond the linear regime, considering situations where the dimensions of the output space and of the ancilla space scale in a non-linear fashion. Motivated by the search of improved bounds one may use in the study of additivity questions, we compute asymptotic expressions for the von Neumann entropies of output matrices for the models under consideration. 

The paper is organized as follows: in Section \ref{sec:review}, we first review the tools available to study moments of outputs of random quantum channels. These techniques were introduced in \cite{collins, collins-sniady} and \cite{collins-nechita-1} and their first applications to quantum information theory were developed in \cite{collins-nechita-1, collins-nechita-3} and \cite{collins-nechita-zyczkowski}. In Section \ref{sec:linear-gen}, we generalize the results of \cite{collins-nechita-3} to the case where the relative dimensions of the input and the output are different. In Section \ref{sec:non-linear}, we generalize the setting of \cite{collins-nechita-3} to the case where the relative dimensions of the input space and the ancilla space have relative polynomial growth. This is motivated by the recent results of \cite{aubrun-szarek-werner-2}, where the authors consider the case $k\sim n^{1/2}$ ($n$ being the dimension of the input/output space and $k$ being the dimension of the ancilla space). We show that depending on the growth, different results occur and that the case where the ancilla space and the input space have the same dimension has a potential for yielding a bigger violation for the additivity of the entropy. Finally, in Section \ref{sec:asympt-ent}, we use these results to provide new bounds for von Neumann entropy of the output of product random quantum channels.

\section{Studying moments of outputs of random quantum channels: techniques and first examples}\label{sec:review}

In this section, we recall, for the convenience of the reader and for the sake of being self-contained, 
techniques to compute the eigenvalue distribution of 
random quantum channels, as well as a few results obtained recently with these techniques. 

The techniques rely on Weingarten calculus (subsection \ref{sub:wg}) and on a graphical model (subsection \ref{sec:planar}). Then, in subsections \ref{sec:k-fixed} and \ref{sec:k-linear}, we recall two applications of these techniques.

\subsection{Weingarten calculus}\label{sub:wg}

In this section, we recall a few facts about the Weingarten calculus, 
useful to evaluate averages  with respect to 
the Haar measure on the unitary group.

\begin{definition}
The unitary Weingarten function 
$\Wg(n,\sigma)$
is a function of a dimension parameter $n$ and of a permutation $\sigma$
in the symmetric group $\S_p$ on $p$ elements, defined as
 the pseudo-inverse of the function $\sigma \mapsto n^{\#  \sigma}$ under the convolution 
for the symmetric group ($\# \sigma$ denotes the number of cycles of the permutation $\sigma$).
\end{definition}

Note that the  function $\sigma \mapsto n^{\# \sigma}$ is invertible for $n\geq p$,
(to see that it is invertible for $n$ large enough, observe that it behaves like $n^p\delta_e$ as $n\to\infty$).
In this case, we can replace the pseudo-inverse by the inverse.
We refer to \cite{collins-sniady} for historical references and further details. 
We shall use the shorthand notation $\Wg(\sigma) = \Wg(n, \sigma)$ when the dimension parameter $n$ is obvious.

The following theorem 
relates integrals with respect to the Haar measure on the unitary group $\U(n)$
and the Weingarten function $\Wg$. 
(see for example \cite{collins-imrn}):

\begin{theorem}
\label{thm:Wg}
 Let $n$ be a positive integer and
$(i_1,\ldots ,i_p)$, $(i'_1,\ldots ,i'_p)$,
$(j_1,\ldots ,j_p)$, $(j'_1,\ldots ,j'_p)$
be $p$-tuples of positive integers from $\{1, 2, \ldots, n\}$. Then
\begin{multline}
\label{bid} \int_{\U(n)}U_{i_1j_1} \cdots U_{i_pj_p}
\overline{U_{i'_1j'_1}} \cdots
\overline{U_{i'_pj'_p}}\ dU=\\
\sum_{\sigma, \tau\in \S_{p}}\delta_{i_1i'_{\sigma (1)}}\ldots
\delta_{i_p i'_{\sigma (p)}}\delta_{j_1j'_{\tau (1)}}\ldots
\delta_{j_p j'_{\tau (p)}} \Wg (n,\tau\sigma^{-1}).
\end{multline}

If $p\neq p'$ then
\begin{equation} \label{eq:Wg_diff} \int_{\U(n)}U_{i_{1}j_{1}} \cdots
U_{i_{p}j_{p}} \overline{U_{i'_{1}j'_{1}}} \cdots
\overline{U_{i'_{p'}j'_{p'}}}\ dU= 0.
\end{equation}
\end{theorem}

We are interested in the values of the Weingarten function in the limit $n \to \iy$. 
The following result encloses all the data we need for our computations
about the asymptotics of the $\Wg$ function; see \cite{collins-imrn} for a proof.

\begin{theorem}\label{thm:mob} For a permutation $\sigma \in \S_p$, 
let $\mathrm{Cycles}(\sigma)$ denote the set of cycles of $\sigma$. Then
\begin{equation}
\Wg (n,\sigma )=(-1)^{n-\# \sigma}
\prod_{c\in \mathrm{Cycles} (\sigma )}\Wg (n,c)(1+O(n^{-2}))
\end{equation}
and 
\begin{equation}
\Wg (n,(1,\ldots ,d) ) = (-1)^{d-1}c_{d-1}\prod_{-d+1\leq j \leq d-1}(n-j)^{-1}
\end{equation}
where $c_i=\frac{(2i)!}{(i+1)! \, i!}$ is the $i$-th Catalan number.
\end{theorem}

The Catalan numbers and $\Wg$ are related to the Moebius function on the lattice of non-crossing partitions, as follows:
\begin{equation}
\Wg(n,\sigma) = n^{-(p + |\sigma|)} (\Mob(\sigma) + O(n^{-2}))
\end{equation}
where $|\sigma |=p-\# \sigma $ is the \emph{length} of $\sigma$, i.e. the minimal number of transpositions that multiply to $\sigma$. We refer to \cite{collins-sniady} for details about the function $\Mob$.

\subsection{Planar expansion}
\label{sec:planar}

The purpose of the graphical calculus introduced in  \cite{collins-nechita-1} is to yield an effective method to evaluate
the expectation of random tensors with respect to the Haar measure on a unitary group. 
%The tensors under consideration can be constructed from a few elementary tensors such as the Bell state, 
%fixed kets and bras, and random unitary matrices.
In graphical language, a tensor corresponds to a \emph{box}, and an appropriate Hilbertian structure yields a correspondence
between boxes and tensors. 
However, the calculus yielding expectations only relies on diagrammatic operations. 

Each box $B$ is represented as a rectangle with decorations on its boundary. The decorations are either white or black, and
belong to $S(B)\sqcup S^*(B)$. 
Figure \ref{fig:box} depicts an example of boxes and diagrams.

\begin{figure}[ht]
\centering
\subfigure[]{\label{fig:box}\includegraphics{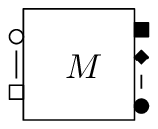}}\qquad\qquad
\subfigure[]{\label{fig:trace}\includegraphics{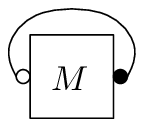}}\qquad\qquad
\subfigure[]{\label{fig:multiplication}\includegraphics{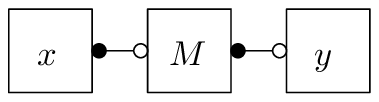}}\\
\subfigure[]{\label{fig:product}\includegraphics{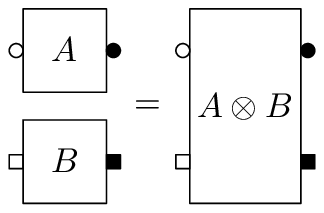}}
\caption{Basic diagrams and axioms: (a) diagram for a general tensor $M$;
 (b) trace of a $(1,1)$-tensor (matrix) $M$; (c) Scalar product $\langle y \;|\; M \;|\; x \rangle$; (d) tensor product of two diagrams. The labels round, square and diamond-shaped labels correspond to pairs of dual finite dimensional complex Hilbert spaces.}
\end{figure}

It is possible to construct new boxes out of old ones by formal algebraic operations such as sums or products.
We call \emph{diagram} a picture consisting in boxes and wires according to the following rule:
a wire may link a white decoration in $S(B)$ to its black counterpart in $ S^*(B)$.
A diagram can be turned into a box by choosing an orientation and a starting point.

Regarding the Hilbertian structure, wires correspond to tensor contractions. 
There exists an involution for boxes and diagrams. It is antilinear and it turns a decoration 
in $S(B)$ into its counterpart in $ S^*(B)$.
Our conventions are close to those of \cite{coecke,jones}.
They should be familiar to 
 the reader acquainted with existing graphical calculus of various types
(planar algebra theory, Feynman diagrams theory, traced category theory).
Our notations are designed to 
fit well to the problem of computing expectations, as shown in the next section. In Figure \ref{fig:trace}, \ref{fig:multiplication} and \ref{fig:product} we depict the trace of a matrix, multiplication of tensors and the tensor product operation.
For details, we refer to \cite{collins-nechita-1}.

The main application of our calculus is to compute expectation of diagrams where some boxes represent random matrices (e.g.
Haar distributed or Gaussian).
For this, we need a concept of \emph{removal} of boxes $U$ and $\ol U$.
A removal $r$ is a way to pair decorations of the $U$ and $\ol U$ boxes appearing in a diagram. 
It therefore consists in  a pairing $\alpha $ of the white decorations of $U$  boxes with the white decorations of $\ol U$ boxes, 
together with a pairing $\beta $ between the black decorations of $U$ boxes and the black decorations of $\ol U$ boxes. 
Assuming that $\D$ contains $p$ boxes of type $U$ and that the boxes $U$ (resp. $\ol U$) are labeled from $1$ to $p$, 
then $r=(\alpha,\beta)$ where $\alpha,\beta$ are permutations of $\mathcal{S}_p$.

Given a removal $r \in \Rem(\D)$, we construct a new diagram $\D_r$ associated to $r$, which has the important property that it no longer contains boxes of type $U$ or $\ol U$. 
One starts by erasing the boxes $U$ and $\ol U$ but keeps the decorations attached to them. 
Assuming that one has labeled the erased boxes $U$ and $\ol U$ with integers from $\{1, \ldots, p\}$, one connects \emph{all} the (inner parts of the) \emph{white} decorations of the $i$-th erased $U$ box with the corresponding (inner parts of the) \emph{white} decorations of the $\alpha(i)$-th erased $\ol U$ box. In a similar manner, one uses the permutation $\beta$ to connect black decorations. 

In \cite{collins-nechita-1}, we proved the following result:
\begin{theorem}\label{thm:Wg_diag}
The following holds true:
\[\E_U(\D)=\sum_{r=(\alpha, \beta) \in \Rem_U(\D)} \D_r \Wg (n, \alpha\beta^{-1}).\]
\end{theorem}

\subsection{Wishart matrices, Marchenko-Pastur distributions and their entropy}\label{sec-application-wishart}

We recall the definition of a \emph{free Poisson} (or Marchenko-Pastur) random variable \cite{MP}. For $c>0$, the probability measure
\[\pi_c=\max (1-c,0)\delta_0+\frac{\sqrt{4c-(x-1-c)^2}}{2\pi x} \; \mathbf{1}_{[1+c-2\sqrt{c},1+c+2\sqrt{c}]}(x) \; dx\]
is called a free Poisson measure of parameter $c$. The plots of the densities for these measures are plotted in Figure \ref{fig:MP-density}.

\begin{figure}[ht]
\centering
\includegraphics[width=0.4\textwidth]{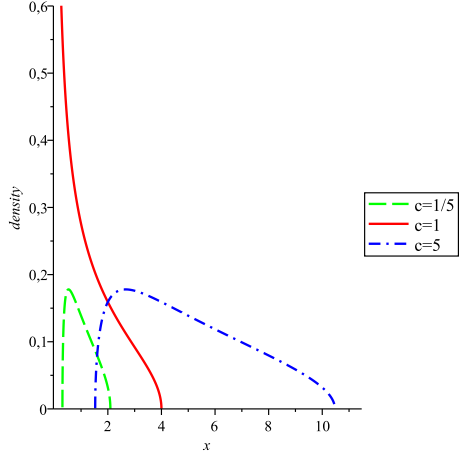}
\caption{Densities for the Marchenko-Pastur measures of parameters $c=1/5$, $c=1$ and $c=5$. For $c=1/5$, only the absolutely continuous part of the measure was plotted; $\pi_{1/5}$ has a Dirac mass of $4/5$ at $x=0$ which is not represented.}
\label{fig:MP-density}
\end{figure}

The free Poisson distribution arises in random matrix theory as the almost sure limit of the eigenvalue counting distribution for Wishart matrices, i.e. matrices $X_nX_n^*$ where $X_n$ is an $n\times \floor{cn}$ matrix whose entries are i.i.d. standard complex Gaussian random variables of variance $1/n$. %%% add ref

From a combinatorial perspective, the free Poisson distribution $\pi_c$ has the nice property that all its \emph{free} cumulants are equal to $c$. Hence, the free moment-cumulant formula (see \cite{nica-speicher}, Lecture 11, pp. 173) reads
\begin{equation}
	\int x^p \; d\pi_c(x) = \sum_{\sigma \in NC(p)} c^{\# \sigma},
\end{equation}
where $\# \sigma$ denotes the number of blocks of the non-crossing partition $\sigma$. From the moment formula, one can obtain the value of the following integral, useful in the computation of von Neumann entropies \cite{page}:
\begin{equation}\label{eq:entropy-free-poisson}
\int_{}^{} x \log x \; d\pi_c(x) = 
\begin{cases}
\frac{1}{2} + c \log c \quad & \text{ if } c \geq 1;\\
\frac{c^2}{2} \quad & \text{ if } 0<c<1.
\end{cases}
\end{equation}

\subsection{Application 1: Fixed ancilla space}\label{sec:k-fixed}

The counterexamples to the additivity conjecture obtained so far arise from the random choice of a quantum channel from the ensemble
\begin{equation}\label{eq:model-k-fixed}
\Phi: \M_{m}\to \M_n
\end{equation}
given by $\Phi(X) = \trace_k[U(X \otimes \ketbra{0}{0})U^*]$,  where $U \in \U(nk)$ is a random unitary matrix and $\ketbra{0}{0}$ is an ancilla rank-one projector. 
The counterexamples rely on the idea of using a tensor product of conjugate channels $\Phi \otimes \ol \Phi$ and, more precisely, the output of this channel when the input is the Bell state ($\{e_1, \ldots, e_m\}$ is some fixed basis of $\C^m$):
\begin{equation}
	 E_m = \frac{1}{m} \sum_{i,j=1}^m \ketbra{e_i}{e_j} \otimes \ketbra{e_i}{e_j}.
\end{equation}
In the regime where $m=tnk$ with $n$ and $k$ integers and $t \in (0,1)$ a fixed parameter, we gave in \cite{collins-nechita-1} 
a complete spectral description of the (random) density matrix $[\Phi \otimes \ol \Phi](E_{tnk})$. The following result improves on the previously known bound of Winter, $\lambda_1 \geq t$:

\begin{theorem}
\label{thm:k_fixed}
Almost surely, as $n \to \iy$, the $k^2$ non-zero eigenvalues of the random matrix $[\Phi\otimes\overline\Phi](E_{tnk}) \in \M_{n^2}(\C)$ converge towards the deterministic probability vector
\[\gamma^{(t)} = \left( t + \frac{1-t}{k^2},\underbrace{\frac{1-t}{k^2}, \ldots, \frac{1-t}{k^2}}_{k^2-1 \text{ times}}\right).\]
\end{theorem}

In the particular case $\Phi: \M_n\to \M_n$, corresponding to $kt=1$, the eigenvalues are
$\frac{1}{k} + \frac{1}{k^2} - \frac{1}{k^3}$, with multiplicity one and $\frac{1}{k^2} - \frac{1}{k^3}$, with multiplicity $k^2-1$. 
Not only the value of the largest eigenvalue is improved from $\frac{1}{k}$ to $\frac{1}{k} + \frac{1}{k^2} - \frac{1}{k^3}$, but the lower eigenvalues are also computed.

A better understanding of the spectrum of the output matrix for the product channel yields immediately better bounds for the Minimum Output Entropy of $\Phi \otimes \ol \Phi$. Applications of this result are twofold. First, it allows for violations of additivity of R\'enyi entropies, for all $p>1$, just by using a qubit as an ancilla state for the input ($t=1/2$). Second, the improvements on the bound for the entropy of the product channel can yield better minimum values of $k$ needed to obtain violations of additivity for $p=1$ (see \cite{fukuda-king-moser, fukuda-king}).

\subsection{Application 2: Ancilla and input space of linear  dimensions}\label{sec:k-linear}

Theorem \ref{thm:k_fixed} is highly non-intuitive: it is not clear why the small eigenvalues should all behave in the same way. Moreover, there was numerical evidence \cite{hayden-winter} that the spectrum should not be flat beyond the second eigenvalue. This raises the question of what happens when the ancilla space is not of fixed dimension, but rather of dimension comparable to the input space. 

The study of such asymptotic regimes was initiated in \cite{collins-nechita-3}. After stating the main result obtained in that paper, we shall generalize in Section \ref{sec:linear-gen} the setting, allowing for the dimension of the input space to vary linearly with the dimension of the output. 

In \cite{collins-nechita-3}, Section 6.3, we considered random quantum channels $\Phi:\M_n(\C) \to \M_n(\C)$ obtained from random Haar unitary matrices $U \in \U(nk)$ via the Stinespring representation
\begin{equation}\label{eq:Stinespring_nk}
	\Phi(X) = \trace_k\left[ U (X \otimes P_k) U^* \right],
\end{equation}
where $P_k \in \M_k(\C)$ is a non-random rank-one projector (pure state) and both $n$ and $k$ grow to infinity at a constant ration $k/n \to c >0$. The diagram for such a channel is represented in Figure \ref{fig:quantum_channel_nk}.

\begin{figure}[ht]
\centering
\includegraphics{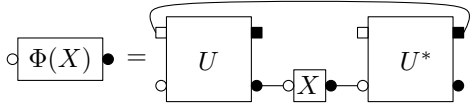}
\caption{Diagram for a quantum channel with equal input and output spaces. The state $P_k$ of the ancilla space is omitted, since it has no role to play in the computations. Round labels attached to boxes correspond to input/output spaces $\C^n$ and square symbols correspond to ancilla spaces $\C^k$.}
\label{fig:quantum_channel_nk}
\end{figure}

For the regime where both $n$ and $k$ grow to infinity at a constant ratio $c$, the main result of \cite{collins-nechita-3} is as follows.

\begin{theorem}
\label{thm:conjugate}
Consider a pair of conjugate random quantum channels $\Phi, \ol \Phi$ in the regime where $n, k \to \iy$, $k \sim cn$. The eigenvalues $\lambda_1 \geq \cdots \geq\lambda_{n^2}$ of the random matrix $Z_{n} = [\Phi\otimes\overline\Phi](E_{n})$ are such that:
\begin{itemize}\addtolength{\itemsep}{-0.5\baselineskip}
\item
The first eigenvalue satisfies $cn\lambda_1\to 1$ (in probability).
\item
The distribution $\frac{1}{n^2-1}\sum_{i=2}^{n^2}\delta_{c^2n^2\lambda_i}$ converges a.s. to a free Poisson distribution of parameter $c^{2}$.
%\item
%Almost surely, 
%\[H(Z_n) = 
%\begin{cases}
%2\log n - \frac{1}{2c^2}+o(1) \quad &\text{ if } \quad c\geq 1,\\
%2\log(cn) - \frac{c^2}{2}+o(1) \quad &\text{ if } \quad 0<c<1,
%\end{cases}\]
%as $n\to\infty$, where $H$ is the von Neumann entropy.
\end{itemize}
\end{theorem}

Here we see that a new phenomenon of two different convergence rates for eigenvalues appears. This is due to the fact that  the $U-\overline U$ model for the product channel contains the ``conjugation'' symmetry. Instead of considering a channel and its complex conjugate, we looked in \cite{collins-nechita-3} at two \emph{independent} quantum channels, taken from the same ensemble. 

\begin{theorem}
\label{thm:independent}
In the regime $k \sim cn$, $n \to \iy$, let $Z_{n} = [\Phi\otimes\Psi](E_{n})$ be the output of the product of two independent quantum channels $\Phi$ and $\Psi$, when the input is a maximally entangled state $E_n$. Then, almost surely, the distribution of the rescaled output matrix $c^2n^2Z_n$ converges towards a free Poisson of parameter $c^2$. 
%Moreover,
%\[H(Z_n) = 
%\begin{cases}
%2\log n - \frac{1}{2c^2}+o(1) \quad &\text{ if } \quad c\geq 1,\\
%2\log(cn) - \frac{c^2}{2}+o(1) \quad &\text{ if } \quad 0<c<1,
%\end{cases}\]
%almost surely as $n\to\infty$, where $H$ is the von Neumann entropy.
\end{theorem}

A striking feature of this asymptotic regime is that the von Neumann entropies of the $U-V$ and $U-\overline U$ models are almost the same. It is then natural to ask whether the $U-\overline U$ symmetry is in fact needed to obtain violations of the additivity. Indeed, it seems that the largest eigenvalue for the output of the product channel does not play a big role in the bounds for the entropy. Having counterexamples with independent channels will be an important step to a better understanding of additivity violations. Possible violations with independent channels would be generic, as opposed to conjugate-channels violations which are not. Also, these considerations give concrete hope that larger violations of additivity could be achieved. 

\section{Generalized linear setting --- input and output spaces of different dimension}\label{sec:linear-gen}

In this section,
we generalize the model of quantum channels we have considered, by allowing input and output spaces of different dimensions. 
We consider random quantum channels $\Phi:\M_m(\C) \to \M_n(\C)$ defined by 
\begin{equation}
	\Phi(X) = \trace_k \left[ U (X \otimes P_l) U^* \right],
\end{equation}
where $l = nk / m$ and $P_l$ is a deterministic rank-one projector in $\M_l(\C)$; we tacitly assume that $l$ is an integer. All three dimensions $m,n$ and $k$ grow to infinity, at constant ratios: $m/n \to b$ and $k/n \to c$, with $b,c\in (0,\infty)$ fixed constants. 
The generalized diagram corresponding to $\Phi$ ias depicted in Figure \ref{fig:quantum_channel_mnk}.

\begin{figure}[ht]
\centering
\includegraphics{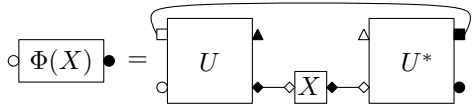}
\caption{Diagram for a quantum channel with different input and output spaces. Round labels attached to boxes correspond to output spaces $\C^n$, square symbols correspond to ancilla spaces $\C^k$ and diamonds correspond to input spaces $\C^m$. The rank-one projector $P_l$ is omitted.}
\label{fig:quantum_channel_mnk}
\end{figure}

When presented with the maximally entangled (or Bell state) $E_{m}\in\M_{m^{2}}(\C)$ as an input, the product conjugate channel $\Phi \otimes \ol \Phi$ produces a random density matrix 
\begin{equation}
	Z = [\Phi \otimes \ol \Phi] (E_{m}) \in \M_{n^2}(\C). 
\end{equation}

The remaining of this section is dedicated to the study of the random matrix $Z$, depicted in Figure \ref{fig:quantum_channel_product_mnk}. The analysis of the spectrum of $Z$ follows closely corresponding results in \cite{collins-nechita-3}, Section 6.3, which is a specialization of this section, in the case $m=n$ (i.e. $b=1$). The spectral properties of the output random matrix $Z$ are summarized in Theorem \ref{thm:mnk}, the main result of this section. The reader in invited to compare the conclusions of Theorems \ref{thm:conjugate} and \ref{thm:mnk}.

\begin{figure}[ht]
\centering
\includegraphics{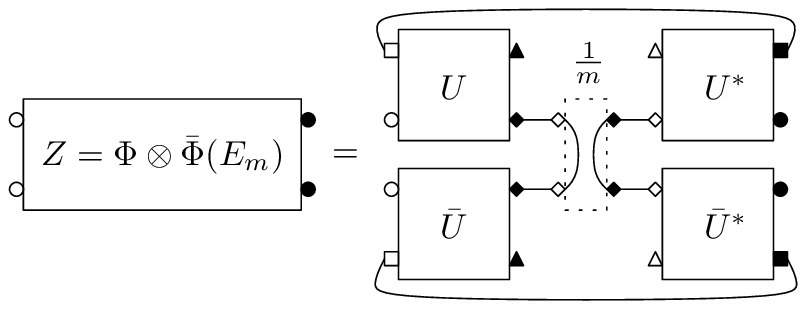}
\caption[d]{Diagram for the output of a product of two conjugate channels, when the input is the maximally entangled state. The complex Hilbert spaces associated to labels are as follows: $\includegraphics{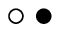}  \leadsto \C^n$, $\includegraphics{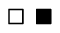} \leadsto \C^k$, $\includegraphics{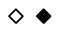} \leadsto \C^m$ and $\includegraphics{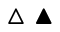} \leadsto \C^l$.}
\label{fig:quantum_channel_product_mnk}
\end{figure}

The first step of our analysis of the output matrix is the computation of the asymptotic moments. 

\begin{proposition}
\label{thm:asympt-Z}
Consider a sequence of random quantum channels $\Phi_{n}$ where $m,n,k \to \iy$, $m/n \to b$ and $k/n \to c$. The asymptotic moments of the output matrix $Z=[\Phi\otimes \ol \Phi ] (E_{m})$ are given by:
\begin{align*}
\trace\left(Z\right) &= 1;\\
\E\trace\left[\left(\frac{c}{b}nZ\right)^2\right] &= 1+ \frac{1}{b^2} + \frac{c^2}{b^2} +o(1);\\
\E\trace\left[\left(\frac{c}{b}nZ\right)^p\right] &= 1+o(1), \quad \forall p \geq 3.
\end{align*}
\end{proposition}
\begin{proof}
The starting point of the proof is the following exact formula for the moments of the random matrix $Z$, obtained from the graph expansion technique detailed in \cite{collins-nechita-1}:
\begin{equation}\label{eq:bi_canal_mnk_exact}
	\E[\trace(Z^p)] = \sum_{\alpha, \beta \in \S_{2p}}k^{\# \alpha}n^{\#(\alpha \gamma^{-1})}
	m^{\# (\beta\delta) - p}\Wg(\alpha\beta^{-1}),
\end{equation}
where the permutations $\gamma, \delta \in \S_{2p}$ are defined as follows. We relabel the index set $\{1, 2, \ldots 2p\}$ as $\{p^B, \ldots, 1^B, 1^T, 2^T, \ldots, p^T\}$ in order to make precise the association of indices with blocks corresponding to the ``top'' channel $\Phi$ ($1^T, \ldots, p^T$) and blocks corresponding to the ``bottom'' channel $\ol \Phi$ ($1^B, \ldots, p^B$). With this notation, we define the permutations
\begin{equation}\label{eq:def-gamma-delta}
\gamma(i^T) = (i-1)^T, \quad \gamma(i^B) = (i+1)^B\quad
\text{ and } \quad
\delta(i^T) = i^B, \quad \delta(i^B) = i^T.
\end{equation}
Using the asymptotic expressions for the dimensions $m \sim bn$, $k \sim cn$ and for the Weingarten function 
(see Theorem \ref{thm:mob}) %note that we did not define \Mob. A faire...
\begin{equation}
	\Wg(\alpha\beta^{-1}) \sim (nk)^{-2p - |\alpha\beta^{-1}|}\Mob(\alpha\beta^{-1}) \sim c^{-2p - |\alpha\beta^{-1}|} n^{-4p - 2|\alpha\beta^{-1}|}\Mob(\alpha\beta^{-1}),
\end{equation}
equation \eqref{eq:bi_canal_mnk_exact} becomes 
\begin{equation}\label{eq:bi_canal_mnk_asympt}
	\E[\trace(Z^p)] \sim \sum_{\alpha, \beta \in \S_{2p}} c^{-(|\alpha|+|\alpha\beta^{-1}|)}b^{p-|\beta\delta|}n^{-f(\alpha,\beta)}\Mob(\alpha\beta^{-1}),
\end{equation}
where the function $f(\alpha, \beta)$ is given by
\begin{equation}
f(\alpha, \beta) = |\alpha|+|\alpha\gamma^{-1}|+|\beta\delta|+2|\alpha\beta^{-1}| - p.
\end{equation}

In order to find the dominating terms in the sums \eqref{eq:bi_canal_mnk_exact} or \eqref{eq:bi_canal_mnk_asympt}, one has to minimize the quantity $f(\alpha, \beta)$ over the permutation group $\S_{2p}$. This has been done in the proof of Theorem 6.8 of \cite{collins-nechita-3}:
\begin{itemize}
\item for $p=1$, $f(\alpha, \beta) \geq 0$, with equality iff. $\alpha=\beta=\id$;
\item for $p=2$, $f(\alpha, \beta) \geq p$, with equality iff. $\alpha=\beta \in \{\id, \delta, \gamma\}$;
\item for $p \geq 3$, $f(\alpha, \beta) \geq p$, with equality iff. $\alpha=\beta=\delta$.
\end{itemize}
One concludes now by plugging the optimal values for the permutations $\alpha$ and $\beta$ into equation \eqref{eq:bi_canal_mnk_asympt}.
\end{proof}

Theorem \ref{thm:asympt-Z} only gives a partial description of the spectrum of the random matrix $Z$; from the moment information one can deduce that there are some eigenvalues on the scale $n^{-1}$ and that the rest of the spectrum is distributed on lower scales, such as $n^{-2}$. Obtaining information on the lower scale eigenvalues 
by brute force
via the moment method is a difficult task, since their asymptotic contribution is negligible with respect to the contribution of the eigenvalue(s) on the scale $n^{-1}$. 
The trick we using to obtain information about the smaller eigenvalues is inspired by Hayden and Winter's proof of the existence of a \emph{large} eigenvalue. Their proof contains, as a byproduct, some information on the eigenvector for the large eigenvalue. We introduce the orthogonal projection $Q = \I - E$, where $E=E_n$
is the maximally entangled state on a product of two copies of the output space $\C^n \otimes \C^n$. Using the (rank $n^2-1$) projector $Q$, we shall obtain some information on the smallest $n^2-1$ eigenvalues of the random output matrix $Z$, by analyzing the ``compressed'' matrix $QZQ$ (which, in a suitable basis, corresponds to considering the $(1,1)$ minor of $Q$).

\begin{proposition}\label{thm:asmpt-QZQ}
Almost surely, the matrix $c^2n^2QZQ$ converges in distribution, to a free Poisson (or Marchenko-Pastur) law of parameter $c^2$.
\end{proposition}

\begin{remark}\label{rem-as}
By `almost surely', we mean with probability one, in any probability space on which the whole sequence (indexed by the input dimension) 
of random quantum channels is defined. In this paper, we supply no proofs of almost sure convergence results, as they require further 
-not so enlightening- technicalities. We refer the interested reader to the appendix of  \cite{collins-nechita-3} for details. Let us just mention
that the proofs rely on the Borel-Cantelli lemma. More precisely, one proves that the covariance of any moment
behaves as $O(n^{-2})$ as the dimension goes to infinity. The fact that $O(n^{-2})$ is summable over $n$ makes it possible to use
the Borel-Cantelli lemma. 
\end{remark}

To prove this result, one can use the arguments of Theorem 6.9 from \cite{collins-nechita-3} \emph{mutatis mutandis}. The method of moments is employed by computing the moments of the random matrix $c^2n^2QZQ$ and showing that they converge to the corresponding moments of the free Poisson distribution of parameter $c^2$:
\begin{equation}
	\lim_{n \to \iy} \frac{1}{n^2}\E[\trace(c^2n^2QZQ)^p] = \int x^p d\pi_{c^2}(x).
\end{equation}

The reader can note that the parameter $b$ (describing the size of the entangled input) has no influence on the lower part of the spectrum. It only appears in the expression of the largest eigenvalue of $Z$, as stated in the following Theorem, which is the main result of this section.

\begin{theorem}
\label{thm:mnk}
Consider a pair of conjugate random quantum channels $\Phi, \ol \Phi$ in the regime where $m, n, k \to \iy$, $m \sim bn$ and $k \sim cn$. The eigenvalues $\lambda_1 \geq \cdots \geq\lambda_{n^2}$ of the random matrix $Z_{n} = [\Phi\otimes\overline\Phi](E_m)$ are such that:
\begin{itemize}\addtolength{\itemsep}{-0.5\baselineskip}
\item
The first eigenvalue satisfies $(c/b)n\lambda_1\to 1$ in probability.
\item
The distribution $\frac{1}{n^2-1}\sum_{i=2}^{n^2}\delta_{c^2n^2\lambda_i}$ converges almost surely to a free Poisson distribution of parameter $c^{2}$.
\end{itemize}
\end{theorem}

The proof of this result combines Theorems \ref{thm:asympt-Z}, \ref{thm:asmpt-QZQ} and Cauchy's interlacing theorem (\cite{bhatia}, Corollary III.1.5). 

\begin{remark}
The almost sure convergence argument described in Remark \ref{rem-as} does not extend to the first item of Theorem \ref{thm:mnk}, as
the covariances tend to zero but are not summable. 
\end{remark}

\section{Non-linear output dimension}\label{sec:non-linear}

In this section we generalize the ``linear'' model of \cite{collins-nechita-3} in a different direction than we did in Section \ref{sec:linear-gen}. We shall consider random quantum channels $\Phi: \M_n(\C) \to \M_n(\C)$ defined by the Stinespring representation \eqref{eq:Stinespring_nk}, where the dimension of the ancilla space $\C^k$ scales with $n$ in a non-linear fashion: 
\begin{equation}
	k \sim cn^d, \qquad \text{ when } n \to \iy.
\end{equation}
Here, $c>0$ and $d \geq 0$ are two real parameters of the model. As before, we are interested in the spectral properties of the random matrix $Z = [\Phi \otimes \ol \Phi] (E_n)$, where $E_n$ is the rank-one projection on the maximally entangled state in the input space of the product channel, $\C^n \otimes \C^n$. 

Before performing a detailed analysis of the spectrum of $Z$, let us make some observations on the role of the parameters $c$ and $d$, as well as on several particular cases already treated in the literature. Generally, in such models, the parameter $c$ will play the role of a \emph{scaling} parameter in the limiting spectral distribution. This phenomenon can be observed in \cite{nechita}, Theorem 5 or in \cite{collins-nechita-3}, Theorem 6.11. On the other hand, the exponent parameter $d$ will have a more \emph{qualitative} role to play, as it will decide the type of behavior of the spectrum of the random matrix $Z$. 

Several particular cases of this very general model of random quantum channels have already been studied in the literature. In \cite{collins-nechita-1}, we studied the case where the dimension $k$ of the ancillary system is fixed, which corresponds in our setting to the case $d=0$ (hence $k=c$). It has been shown (Section \ref{sec:k-fixed}, Theorem \ref{thm:k_fixed}) that in this situation, the random matrix $Z$ has one large eigenvalue ($\lambda_1 = k^{-1} + k^{-2} - k^{-3})$) and that the rest of the spectrum is ``flat'': $\lambda_2 = \cdots = \lambda_{k^2} = k^{-2} - k^{-3}$. In \cite{collins-nechita-3}, the case $d=1$ was investigated; this corresponds to a coupling with an ancilla space of dimension $k$ which scales as $k \sim cn$. The situation was rather different in this case: one large eigenvalue of size $(cn)^{-1}$ was observed, and the lower spectrum was not flat, having a free Poisson $\pi_{c^2}$ shape. As a final remark, note that the model under study here is different than the one in Section \ref{sec:linear-gen}, where inputs of arbitrary size were considered.

The main result of this section is the following theorem, which classifies the spectral behavior of the output random matrix $Z$ in terms of the parameter $d$. 

\begin{proposition}\label{thm:non-linear}
The asymptotic moments of the random output matrix $Z = [\Phi \otimes \ol \Phi] (E_n)$ are given by:
\begin{enumerate}
	\item If $d=0$ (see \cite{collins-nechita-1}):	
		\begin{equation}
			\E[\trace(Z^p)]  =  \left( \frac{1}{c} + \frac{1}{c^2} - \frac{1}{c^3} \right)^p + (c^2-1) \left(\frac{1}{c^2}-\frac{1}{c^3} \right)^p +o(1) \qquad \forall p \geq 2.
		\end{equation}
	\item If $d\in (0, 1)$:
		\begin{equation}
			\begin{split}
				&\E[\trace(Z^p)]  \sim  2c^{-2}n^{-2d} \sim 2k^{-2}  \qquad \text{if } p = 2; \\
				&\E[\trace(Z^p)]  \sim  c^{-p}n^{-dp} \sim k^{-p}  \qquad \forall p \geq 3. \\
			\end{split}
		\end{equation}	
	\item If $d=1$ (see \cite{collins-nechita-3}):
		\begin{equation}
			\begin{split}
				&\E[\trace(Z^p)]   \sim (1+2c^{-2})n^{-2} \qquad \text{if } p = 2; \\
				&\E[\trace(Z^p)]   \sim  c^{-p}n^{-dp} \sim k^{-p} \qquad \forall p \geq 3. \\
			\end{split}
		\end{equation}
	\item If $d \in (1,2)$:
		\begin{equation}
			\begin{split}
				&\E[\trace(Z^p)]   \sim  n^{-(2p-2)}  \qquad \text{if } p < \frac{2}{2-d}; \\
				&\E[\trace(Z^p)]   \sim  (1+c^{-p})n^{-dp}  \sim  (1+c^{-p})n^{-(2p-2)}  \qquad \text{if } p = \frac{2}{2-d}; \\
				&\E[\trace(Z^p)]   \sim  c^{-p}n^{-dp}   \sim k^{-p}  \qquad \text{if } p > \frac{2}{2-d}. \\
			\end{split}
		\end{equation}	
	\item If $d \geq 2$:
		\begin{equation}
			\begin{split}
				&\E[\trace(Z^p)]   \sim n^{-(2p-2)}  \qquad \forall p \geq 2.
			\end{split}
		\end{equation}		
\end{enumerate}
\end{proposition}
\begin{proof}
The starting point of the proof is the following exact moment formula, obtained via the graphical calculus introduced in \cite{collins-nechita-1} 
%mettre une reference interne
and the Weingarten formula:
\begin{equation}
		\E[\trace(Z^p)] = \sum_{\alpha, \beta \in \S_{2p}}k^{\# \alpha}n^{\#(\alpha \gamma^{-1})+\# (\beta\delta) - p}\Wg(\alpha\beta^{-1}).
\end{equation}
After using the asymptotic expression for the ancillary dimension $k \sim cn^d$ and the Weingarten function 
\begin{equation}
\begin{split}
\Wg(\alpha\beta^{-1}) &\sim (nk)^{-2p - |\alpha\beta^{-1}|2p - |\alpha\beta^{-1}|}\Mob(\alpha\beta^{-1}) \\
&\sim c^{-2p - |\alpha\beta^{-1}|}n^{-(d+1)(2p + |\alpha\beta^{-1}|)}\Mob(\alpha\beta^{-1}),
\end{split}
\end{equation}
we obtain the following expression (which holds if the right-hand-side is non-zero):
\begin{equation}\label{eq:non-linear-asympt}
	\E[\trace(Z^p)] \sim \sum_{\alpha, \beta \in \S_{2p}} c^{-(|\alpha|+|\alpha\beta^{-1}|)}n^{-S(\alpha, \beta)}\Mob(\alpha\beta^{-1}),
\end{equation}
where $S(\alpha, \beta)$ is the exponent of $n$ in the preceding sum:
\begin{equation}
	S(\alpha, \beta) = d|\alpha|+|\alpha\gamma^{-1}| + |\beta\delta|+(d+1)|\alpha\beta^{-1}| - p.
\end{equation}
In this proof, since the permutation $\delta$ is a product of disjoint transpositions, thus an involution, we shall use the fact that $|\alpha\delta|=|\alpha\delta^{-1}|=|\alpha^{-1}\delta|$. In order to find the dominating terms in the sum \eqref{eq:non-linear-asympt}, one has to minimize the quantity $S(\alpha, \beta)$ over $\S_{2p}^2$. The solution to this problem is obtained in three steps. First, using the following triangular inequalities:
\begin{equation}\label{eq:alpha-beta}
  \begin{split}
	  |\alpha|+|\alpha\beta^{-1}| &\geq |\beta|,\\
		|\alpha\gamma^{-1}|+|\alpha\beta^{-1}| &\geq |\beta \gamma^{-1}|,
	\end{split}
\end{equation}
we obtain that
\begin{equation}\label{eq:S1}
	S(\alpha, \beta) \geq d|\beta|+|\beta\gamma^{-1}| + |\beta\delta| - p =: S_1(\beta),
\end{equation}
where the inequality can be saturated if , e.g. $\alpha = \beta$. The minimization problem for the $S_1$ function, although much simpler that the two-variable problem for $S$, can be further simplified by replacing the two-cycle permutation $\gamma$ with the full-cycle:
\begin{equation}\label{eq:S2}
	\text{minimize } \qquad S_2(\beta) = d|\beta|+|\beta\gammat^{-1}| + |\beta\delta| - p,
\end{equation}
where $\gammat = (p^T \cdots 2^T \; 1^T \; 1^B \; 2^B \cdots p^B)$. Since the permutation $\delta$ is an element of the geodesic $\id \to \gammat$, the solution for the $S_2$ minimization problem is easy to find using one or more of the following inequalities:
\begin{equation}
  \begin{split}
  	|\beta| &\geq 0;\\
	  |\beta|+|\beta \gammat^{-1}| &\geq |\gammat| = 2p-1;\\
		|\beta \delta| + |\beta \gammat^{-1}| &\geq |\delta \gammat^{-1}| = p-1.
	\end{split}
\end{equation}
The solution for the $S_2$ problem, in terms of the value of the parameter $d$, is summarized in Table \ref{tab:S2}.
\begin{table}[ht]
	\centering
		\begin{tabular}{|c|c|c|}
			\hline
			$d$ & Minorant for $S_2(\beta)$ & Equality cases\\
			\hline\hline
			$0$ & $-1$ & $\{\delta \to \gammat\}$ \\ \hline
			$(0,2)$ & $dp-1$ & $\delta$ \\ \hline
			$2$ & $2p-1$ & $\{\id \to \delta\}$ \\ \hline
			$(2, \iy)$ & $2p-1$ & $\id$ \\ \hline
		\end{tabular}
	\caption{Solution to the $S_2$ minimization problem of equation \eqref{eq:S2}.}
	\label{tab:S2}
\end{table}

Next, we move towards finding the minimum of $S_1(\beta)$, defined in equation \eqref{eq:S1}. The permutations $\gamma$ and $\gammat$ are at distance one:
\begin{equation}
	\gamma = \gammat \cdot (1^B\; p^T),
\end{equation}
hence the same holds for $\beta\gamma^{-1}$ and $\beta\gammat^{-1}$. We have thus $S_1(\beta) = S_2(\beta) \pm 1$ and, even more precisely, 
\begin{equation}
  S_1(\beta) = 
		\begin{cases}
		S_2(\beta) - 1 \quad &\text{if } 1^B \text{ and } p^T \text{ are in the same block of } \beta\gammat^{-1}, \\
		S_2(\beta) + 1 \quad &\text{otherwise.}
		\end{cases}	
\end{equation}

In the same manner that is was argued in \cite{collins-nechita-3}, $p^T$ and $1^B$ belong to the same block of $\beta\gammat^{-1}$ if and only if $\beta \leq \gamma$ (the permutations being compared with the partial order relation on the corresponding non-crossing partitions). Analyzing the different equality cases in Table \ref{tab:S2} and using the fact that the unique element of the geodesic set $\{\id \to \delta\}$ which is smaller than $\gamma$ is $\beta=\id$, we conclude that, for $d \geq 2$, $S_1(\beta) \geq 2p-2$, with equality iff. $\beta = \id$. The case $d \in [0,2)$ is more intricate, since one cannot have $S_1(\beta) = S_2(\beta) - 1$ and saturate at the same time the lower bound for $S_2$. 

For $d=0$, it was shown in \cite{collins-nechita-1} that $S_1(\beta) \geq 0$, with equality iff $\beta \in \{\delta, \gamma\}$. Choose $d \in (0, 1]$ and consider $\beta \in \S_{2p}$ such that $S_1(\beta) = S_2(\beta) - 1$. Then, since $\beta \leq \gamma$, one has $|\beta\delta|\geq p$ and thus
\begin{equation}
	\begin{split}
		S_1(\beta) &= S_2(\beta) - 1 = d|\beta|+|\beta\gamma^{-1}| + |\beta\delta| - p - 1\\
		&= d(|\beta|+|\beta\gamma^{-1}|) + (1-d)|\beta\gamma^{-1}| + |\beta\delta| - p - 1\\
		&\geq d(2p-1) + (1-d) + p - p - 1 = d(2p - 2).
	\end{split}
\end{equation}
We conclude that, for $d \in (0, 1]$, $p \geq 3$ and $\beta$ such that $S_1(\beta) = S_2(\beta) - 1$, $S_1(\beta) > S_1(\delta) = dp$. Thus, for $d \in (0, 1]$ and $p \geq 3$, the minimum $dp$ is attained only at the point $\beta = \delta$. For $p=2$, an exhaustive search in $S_4$ reveals that, for $d \in (0,1)$, $S_1(\beta) \geq dp$, with equality iff $\beta = \delta, \gamma$, and, for $d =1$, $S_1(\beta) \geq p$, with equality iff $\beta = \id, \delta, \gamma$.

In the case $d \in (1,2)$, the situation is different. We shall consider three cases:
\begin{enumerate}
	\item $\beta \notin \{\id \to \gammat\}$. Since $\beta$ is not a geodesic permutation, we have that $|\beta|+|\beta\gammat^{-1}| \geq 2p+1$, $|\beta|+|\beta\delta| \geq p+2$, $|\beta| \geq 1$ and $|\beta\delta| \geq 1$. It follows that 
		\begin{equation}
			\begin{split}
				S_1(\beta) &\geq S_2(\beta)-1 \geq |\beta|+|\beta\gammat^{-1}| + (d-1)(|\beta| + |\beta\delta|) + (2-d)|\beta\delta| - p - 1 \\
				&\geq 2p+1 + (d-1)(p+2) + (2-d) - p - 1 = dp + d > dp = S_1(\delta).
			\end{split}
		\end{equation}
	\item $\beta \in \{\id \to \gamma\}$. In this situation, $S_1(\beta) = S_2(\beta)-1$ and thus (we use the fact that $|\beta \delta|\geq p$ in this case)
		\begin{equation}
			\begin{split}
				S_1(\beta) &= S_2(\beta)-1 \geq |\beta|+|\beta\gammat^{-1}| + (d-1)|\beta| + |\beta\delta| - p - 1 \\
				&\geq 2p-1 + 0 + p - p - 1 \geq 2p -2 = S_1(\id),
			\end{split}
		\end{equation}	
		with equality iff. $\beta=\id$.
	\item $\beta \in \{\id \to \gammat\}, \; \beta \notin \{\id \to \gamma\}$. In this situation, $S_1(\beta) = S_2(\beta)+1$ and thus
		\begin{equation}
			\begin{split}
				S_1(\beta) &= S_2(\beta)+1 \geq |\beta|+|\beta\gammat^{-1}| + (d-1)(|\beta| + |\beta\delta|) + (2-d)|\beta\delta| - p + 1 \\
				&\geq 2p-1 + (d-1)p + (d-2)0 - p + 1 \geq dp = S_1(\delta),
			\end{split}
		\end{equation}	
		with equality iff. $\beta=\delta$.
\end{enumerate}
Analyzing the three cases, we conclude that if $p \leq 2/(2-d)$, we have $S_1(\beta) \geq 2p-2$, otherwise $S_1(\beta) \geq dp$, the unique minimizers being respectively $\beta=\id$, $\beta=\delta$ and $\beta \in \{\id, \delta\}$ for the interface case $p = 2/(2-d)$. The answer to the $S_1$ minimization problem is summarized in Table \ref{tab:S1}.
\begin{table}[ht]
	\centering
		\begin{tabular}{|c|c|c|}
			\hline
			$d$ & Minorant for $S_1(\beta)$ & Equality cases\\
			\hline\hline
			$0$ & $0$ & $\{\delta \to \gamma\}$ \\ \hline
			$(0,1) \quad (p = 2)$ & $2d$ & $\delta, \gamma$ \\ \hline
			$1 \quad (p = 2)$ & $2$ & $\id, \delta, \gamma$ \\ \hline
			$(0,1] \quad (p\geq 3)$ & $dp$ & $\delta$ \\ \hline
			$(1,2) \quad (p < \frac{2}{2-d})$ & $2p-2$ & $\id$ \\ \hline
			$(1,2) \quad (p = \frac{2}{2-d})$ & $2p-2 = dp$ & $\id, \delta$ \\ \hline				
			$(1,2) \quad (p > \frac{2}{2-d})$ & $dp$ & $\delta$ \\ \hline
			$[2, \iy)$ & $2p-2$ & $\id$ \\ \hline
		\end{tabular}
	\caption{Solution to the $S_1$ minimization problem of equation \eqref{eq:S1}.}
	\label{tab:S1}
\end{table}

In order to solve the initial problem for the two-variable function $S(\alpha, \beta)$, one needs to notice that, for a fixed value of $\beta$, the only possibility to saturate the inequalities \eqref{eq:alpha-beta} at the same time is $\alpha = \beta$. Since, for $d >0$, both these inequalities are used to go from $S$ to $S_1$, we conclude that for strictly positive $d$, the solution for the $S$ minimization problem can be found in Table \ref{tab:S1}, with $\alpha = \beta$. The solution for $d=0$ has been entirely described in \cite{collins-nechita-1}. One concludes by replacing the values for the minimizing permutations into equation \eqref{eq:non-linear-asympt}.
\end{proof}

In the cases $d=0$ and $d\geq 2$, the behavior of the eigenvalues can be easily deduced from the moment information, since one can identify in the formulas the moments of some probability distribution. In the case $d=0$, as it was argued in \cite{collins-nechita-1}, the larges eigenvalue of $Z$ converges to $\frac{1}{c} + \frac{1}{c^2} - \frac{1}{c^3}$ and the $k^2-1$ others converge to $\frac{1}{c^2} - \frac{1}{c^3}$. Such a behavior is typical for the model we study, with a spectrum containing one large eigenvalue and $k^2-1$ (or $n^2-1$) identical smaller eigenvalues. The case $d \geq 2$ is somehow atypical: all the $n^2$ eigenvalues behave like $n^{-2}$. This is due to the fact that the ``large'' eigenvalue, which usually behaves as $k^{-1}$ has no contribution asymptotically. We summarize these facts in the following proposition.

\begin{proposition}
In the regime $d=0$ (which corresponds to considering a fixed ancilla dimension $k=c$), the eigenvalues of the random matrix $Z$ are such that, almost surely, in the limit $n \to \iy$,
\begin{equation}
	\lambda_1  \to \frac{1}{c} + \frac{1}{c^2} - \frac{1}{c^3} \qquad \text{and} \qquad \lambda_2, \ldots, \lambda_{k^2} \to \frac{1}{c^2} - \frac{1}{c^3}.
\end{equation}
In the regime $d\geq 2$, the (rescaled) empirical spectral distribution of $Z$ converges to the Dirac mass at 1, $\delta_1$:
\begin{equation}
	\lim_{n\to \iy} \frac{1}{n^2} \sum_{i=1}^{n^2} \delta_{n^2 \lambda_i} = \delta_1.
\end{equation}

\end{proposition}

As in Section \ref{sec:linear-gen}, in order to understand fully the eigenvalues of the random matrix $Z$, we need to understand the lower part of the spectrum, in the remaining cases $d \in (0,1)$ and $d \in (1,2)$ (the case $d=1$ being treated in \cite{collins-nechita-3}). This is done by ``pinching'' the matrix $Z$ by the projector $Q = \I - E_n$, orthogonal to the maximally entangled state $E_n$.

\begin{proposition}\label{prop:asymptotics-QZQ-01}
The matrix $k^2QZQ$ converges, in moments, to the Dirac mass at 1, $\delta_1$.
\end{proposition}
\begin{proof}
We follow the idea of the proof of Theorem 6.10 in \cite{collins-nechita-3}, compute the moments of the rank $k^2$ matrix $k^2QZQ$, and show that they converge to the corresponding moments of the limit law:
\begin{equation}
	\lim_{n \to \iy} \frac{1}{k^2}\E[\trace(k^2QZQ)^p] = 1.
\end{equation}
After replacing $Q=\I-E$ and developing the product, we get
\begin{equation}\label{eq:QZQ-exact}
	\begin{split}
		\frac{1}{k^2}\E[\trace(k^2QZQ)^p] &= k^{2p-2}\E[\trace (\I-E)Z(\I-E)Z\cdots(\I-E)Z] \\
&= c^{2p}n^{2p-2} \sum_{f \in \F_p} (-1)^{|f^{-1}(E)|}n^{-|f^{-1}(E)|} \E[\trace f(1) Z f(2)Z \cdots f(p) Z],
	\end{split}
\end{equation}
where $\F$ is a set of the $2^p$ choice functions $f:\{1, 2, \ldots, p\} \to \{\I, E\}$. The factor $n^{-|f^{-1}(E)|}$ is due to the normalization of the Bell states $E = E_n$. The moment $\E[\trace f(1) Z f(2)Z \cdots f(p) Z]$ is computed via the graphical Weingarten calculus:
\begin{equation}
	\E[\trace f(1) Z f(2)Z \cdots f(p) Z] = \sum_{\alpha, \beta \in \S_{2p}}k^{\# \alpha}n^{\#(\alpha \fhat^{-1}) +\# (\beta\delta) - p}\Wg(\alpha\beta^{-1}),
\end{equation}
where $\fhat \in \S_{2p}$ is the permutation associated to the choice function $f \in \F_p$ describing the way $f$ connects the different instances of the channel (the arithmetic operations of indices $i$ should be understood modulo $p$):
\begin{align*}
i^T &\stackrel{\fhat}{\mapsto}
\begin{cases}
(i-1)^T \quad &\text{ if } f(i)=\I\\
i^B \quad &\text{ if } f(i)=E,
\end{cases}
\\
i^B &\stackrel{\fhat}{\mapsto}
\begin{cases}
(i+1)^B \quad &\text{ if } f(i+1)=\I\\
i^T \quad &\text{ if } f(i+1)=E.
\end{cases}
\end{align*}
 
Exactly as in the proof of Theorem 6.10 from \cite{collins-nechita-1}, one has to understand the possible cancellations of high powers in $n$. In order to do this, we rewrite the non-asymptotic equation \eqref{eq:QZQ-exact} as
\begin{equation}
	\frac{1}{k^2}\E[\trace(k^2QZQ)^p] = \sum_{\alpha, \beta \in \S_{2p}}k^{4p-|\alpha|-2} n^{3p-|\beta \delta|} \Wg(\alpha \beta^{-1}) \sum_{f \in \F_p} (-1)^{|f^{-1}(E)|} n^{-(|f^{-1}(E)| + |\alpha \fhat^{-1}|)}.
\end{equation}
As in \cite{collins-nechita-1}, we can show that for all permutations $\alpha \in \V$, the sum over all choices $f \in \F_p$ is exactly zero, where
\begin{equation}
	\begin{split}
		\V &= \{\sigma \in \S_{2p} \; | \; \exists i \in \{1, \ldots, p\} \text{ s.t. } \sigma(i^T) = i^B \text{ or } \sigma(i^B) = i^T \} \\
		&= \{\sigma \in \S_{2p} \; | \; \sigma \delta \text{ has at least one fixed point}\}.
	\end{split}
\end{equation}
Hence, 
\begin{equation}\label{eq:QZQ-asympt}
	\frac{1}{k^2}\E[\trace(k^2QZQ)^p] \sim \!\!\!\!\!\!\!\!\sum_{f \in \F_p, \; \alpha, \beta \in \S_{2p}, \; \alpha \notin \V}\!\!\!\!\!\!\!\! (-1)^{|f^{-1}(E)|} c^{2p-2-(|\alpha| + |\alpha \beta^{-1}|)} n^{-S(\alpha, \beta, f)} \Mob(\alpha \beta^{-1}),
\end{equation}
where
\begin{equation}
	S(\alpha, \beta, f) = |\beta \delta| + d|\alpha| + (d+1)|\alpha \beta^{-1}|+|f^{-1}(E)| + |\alpha \fhat^{-1}| - p(2d+1)+2d .
\end{equation}
Since $\alpha \notin \V$, $\alpha \delta$ has no fixed point, and hence $|\alpha \delta| \geq p$. Using the facts that $|\alpha \beta^{-1}| + |\beta \delta| \geq |\alpha \delta|$, $|\alpha \beta^{-1}| \geq 0$ and $|\alpha|+|\alpha \fhat^{-1}| \geq |\fhat|$, we obtain that
\[ S(\alpha, \beta, f) \geq |f^{-1}(E)| + d|\fhat| + (1-d)|\alpha \fhat^{-1}| - 2dp + 2d,\]
with equality if and only if $\alpha = \beta \in \{\id \to \fhat\}$. 

The number of cycles of $\fhat$ is easily shown to be:
\begin{equation}
	\# \fhat = 
\begin{cases}
2 \quad &\text{ if } f \equiv \I,\\
|f^{-1}(E)| \quad &\text{ otherwise}.
\end{cases}
\end{equation}
hence, $|f^{-1}(E)| + d|\fhat| \geq 2dp - 2d$, with equality iff $f \equiv \I$. We can conclude that $S(\alpha, \beta, f) \geq 0$, with equality if and only if $\alpha = \beta = \hat{\I} = \gamma$. Replacing these values in equation \eqref{eq:QZQ-asympt}, we obtain the announced result.
\end{proof}

Using Cauchy's interlacing theorem (\cite{bhatia}, Corollary \text{III.1.5}) for the eigenvalues of $QZQ$ and those of $Z$, we obtain the complete description of the spectrum of the random output matrix.

\begin{theorem} \label{thm:asymptotics-Z-01}
	In the regime $k \sim cn^d$, with $d \in (0, 1)$, the eigenvalues $\lambda_1\geq \cdots \geq\lambda_{n^2}$ of $Z = [\Phi \otimes \ol \Phi] (E_n)$ satisfy:
	\begin{itemize}
		\item In probability, $k\lambda_1\to 1$.
		\item Almost surely, $\frac{1}{k^2-1}\sum_{i=2}^{k^2}\delta_{k^2\lambda_i}$ converges in distribution to the Dirac mass at 1, $\delta_1$.
		\item The remaining $n^2-k^2$ eigenvalues are null: $\lambda_{k^2+1} = \cdots = \lambda_{n^2} = 0$.
	\end{itemize}
\end{theorem}

We finally treat the regime $d \in (1,2)$, stating the results and underlining the (small) differences between the proofs in this case and the proofs of Proposition \ref{prop:asymptotics-QZQ-01} and, respectively, of Theorem \ref{thm:asymptotics-Z-01}.

\begin{proposition}\label{prop:asymptotics-QZQ-12}
In the regime $d \in (1,2)$, the matrix $n^2QZQ$ converges, in moments, to the Dirac mass at 1, $\delta_1$.
\end{proposition}
\begin{proof}
 
As in the proof of Proposition \ref{prop:asymptotics-QZQ-01}, we obtain
\begin{equation}\label{eq:QZQ-asympt-12}
	\frac{1}{n^2}\E[\trace(n^2QZQ)^p] \sim \!\!\!\!\!\!\!\!\sum_{f \in \F_p, \; \alpha, \beta \in \S_{2p}, \; \alpha \notin \V}\!\!\!\!\!\!\!\! (-1)^{|f^{-1}(E)|} c^{-(|\alpha| + |\alpha \beta^{-1}|)} n^{-S(\alpha, \beta, f)} \Mob(\alpha \beta^{-1}),
\end{equation}
where
\begin{equation}
	S(\alpha, \beta, f) = |\beta \delta| + d|\alpha| + (d+1)|\alpha \beta^{-1}|+|f^{-1}(E)| + |\alpha \fhat^{-1}| - 3p+2.
\end{equation}
Using the same arguments as before, we obtain that
\[ S(\alpha, \beta, f) \geq |f^{-1}(E)| + |\fhat| + (d-1)|\alpha| - 2p + 2,\]
with equality if and only if $\alpha = \beta \in \{\id \to \fhat\}$. 

Counting the number of cycles of $\fhat$, we have that $|f^{-1}(E)| + |\fhat| \geq 2p - 2$, with equality iff $f \equiv \I$. We  conclude that $S(\alpha, \beta, f) \geq 0$, with equality if and only if $\alpha = \beta = \id$. The result follows by plugging the minimizing values into the asymptotic expression \eqref{eq:QZQ-asympt-12}.
\end{proof}

\begin{theorem} \label{thm:asymptotics-Z-12}
	In the regime $k \sim cn^d$, with $d \in (1, 2)$, the eigenvalues $\lambda_1\geq \cdots \geq\lambda_{n^2}$ of $Z = [\Phi \otimes \ol \Phi] (E_n)$ satisfy:
	\begin{itemize}
		\item In probability, $k\lambda_1\to 1$.
		\item Almost surely, $\frac{1}{n^2-1}\sum_{i=2}^{n^2}\delta_{n^2\lambda_i}$ converges in distribution to the Dirac mass at 1, $\delta_1$.
	\end{itemize}
\end{theorem}

\section{Asymptotics of the von Neumann entropy}\label{sec:asympt-ent}

Using the moment information and the behavior of the lower part of the spectrum deduced in Section \ref{sec:non-linear}, we now analyze the von Neumann entropy of the random output matrix $Z$:
\begin{equation}
	H(Z) = -\sum_{i=1}^{n^2} \lambda_i \log \lambda_i.
\end{equation}

In the physical literature, the idea of using conjugate quantum channels to tackle important questions, such as the additivity of the minimum output entropy for quantum channels, dates back to the work of A. Winter and P. Hayden \cite{hayden-winter}. To bound the entropy of $Z$, they use the following fact, coming from linear algebra, which is independent of the random model used: the largest eigenvalue of $Z$ is larger that the inverse of the dimension of the ancilla space, $\lambda_1 \geq 1/k$. This bound (which is actually a bound on the operator norm of $Z$) yields the following bound on the von Neumann entropy:
\begin{equation}\label{eq:naive-bound}
	H(Z) \leq 2 \log k - \frac{\log k}{k} + \frac{1}{k}.
\end{equation}

One of the main applications of the results in this paper is the fact that our exact spectral information yield better upper bounds for the von Neumann entropy in some specific cases. 

\begin{theorem}\label{thm:entropy}
The asymptotic von Neumann entropy of the random output matrix \\
$Z = [\Phi \otimes \ol \Phi] (E_{n})$ is given by:
\begin{enumerate}
	\item If $d=0$ ($k=c$ is an integer):	
		\begin{equation}\label{eq:entropy-d-0}
			H(Z)  =  -\left( \frac{1}{c} + \frac{1}{c^2} - \frac{1}{c^3} \right)\log \left( \frac{1}{c} + \frac{1}{c^2} - \frac{1}{c^3} \right) - (c^2-1) \left(\frac{1}{c^2}-\frac{1}{c^3} \right) \log \left(\frac{1}{c^2}-\frac{1}{c^3} \right) +o(1).
		\end{equation}
	\item If $d\in (0, 1)$:
		\begin{equation}
			\begin{split}
				&H(Z) = 2 \log k + o(1).
			\end{split}
		\end{equation}	
	\item If $d=1$ (see \cite{collins-nechita-3}):
		\begin{equation}\label{eq:entropy-d-1}
			H(Z) = 
				\begin{cases}
				2\log k - \frac{c^2}{2}+o(1) \quad &\text{ if } \quad 0<c<1,\\
				2\log n - \frac{1}{2c^2}+o(1) \quad &\text{ if } \quad c\geq 1.
				\end{cases}
		\end{equation}
	\item If $d \in (1,2)$:
		\begin{equation}
			\begin{split}
				&H(Z) = 2 \log n + o(1).
			\end{split}
		\end{equation}	
	\item If $d \geq 2$:
		\begin{equation}
			\begin{split}
				&H(Z) = 2 \log n + o(1).
			\end{split}
		\end{equation}		
\end{enumerate} 
\end{theorem}

One can make the following very instructive observations about the above theorem. First, notice that in all cases where $d>0$, the main contribution to the von Neumann entropy is given by the lower part of the spectrum, and not by the main eigenvalue. This is in contrast with the case of the $p$-R\'enyi entropies ($p>0$), where the largest eigenvalue gives the main contribution \cite{hayden-winter}. 

Another important point concerns the fixed dimension case $d=0$  and the  linear case, $d=1$, $k \sim cn$. In these regimes, the entropy defect $2 \log (k \wedge n) - H(Z)$ is macroscopic \eqref{eq:entropy-d-0},\eqref{eq:entropy-d-1}. 
This improves considerably the naive bound \eqref{eq:naive-bound} and may provide more insight into the question of additivity of minimal output entropies, as argued in \cite{hastings, collins-nechita-2, aubrun-szarek-werner-2}.
However, Hastings' techniques \cite{hastings}, as well as the recent developments of Aubrun, Szarek and Werner \cite{aubrun-szarek-werner-2} do not seem to apply to this linear regime. Studying additivity violations in the linear regime and using the bounds in Theorem \ref{thm:entropy} to provide larger violations remain interesting open problems.

%%%%%%%%%%%%%%%%%%%%%%%%%%%%%%%%%%%%%%%%%%%%%%%%%%%%%%%%%%%%
\section*{Acknowledgements}

B.C. was partly funded by ANR GranMa and ANR Galoisint.
The research of both authors was supported in part by NSERC
grant RGPIN/341303-2007.

%==========================================================

%==========================================================
% Back Matter (References and Notes)
%----------------------------------------------------------
% Style and layout of the references
\bibliographystyle{mdpi}
\makeatletter
\renewcommand\@biblabel[1]{#1. }
\makeatother
%----------------------------------------------------------
% Use the following option to include external BibTeX files:
%\bibliography{template}
%----------------------------------------------------------

\end{document}